\documentclass[11pt]{article}
\usepackage{graphicx}
\usepackage[margin=1.1in]{geometry}
\usepackage{bm}
\usepackage{subfig}
\usepackage{natbib}
\usepackage{amsmath,amssymb,amsfonts,amsthm}

\newcommand{\N}{\ensuremath{\mathbb{N}}}

\newcommand{\PP}{\ensuremath{\mathbb{P}}}
\newcommand{\E}{\ensuremath{\mathbb{E}}}

\newcommand{\A}{\ensuremath{h(p^{\ast})}}
\newtheorem{theorem}{Theorem}
\newtheorem{lemma}{Lemma}
\newtheorem{corollary}{Corollary}
\newtheorem{algorithm}{Algorithm}
\newtheorem{condition}{Condition}

\theoremstyle{remark}

\newtheorem{example}{Example}

\newcommand{\Supp}[1]{{Supporting Information}}

\title{MMCTest -- A Safe Algorithm for Implementing Multiple Monte Carlo Tests}
\author{Axel Gandy and Georg Hahn
  \\Department of Mathematics, Imperial College London}
\date{}

\begin{document}
\maketitle

\begin{abstract}
  Consider testing multiple hypotheses using tests that can only be
  evaluated by simulation, such as permutation tests or bootstrap
  tests.  This article introduces \texttt{MMCTest}, a sequential
  algorithm which gives, with arbitrarily high probability,
  the same classification as a specific multiple testing procedure applied to ideal p-values.
  The method can be used with a class of multiple testing procedures
  which includes the Benjamini \& Hochberg False Discovery Rate
  (FDR) procedure and the Bonferroni correction
  controlling the Familywise Error Rate.
  One of the key features of the algorithm is that it stops
  sampling for all the hypotheses which can already be decided
  as being rejected or non-rejected.
  \texttt{MMCTest} can be interrupted at any stage and then
  returns three sets of hypotheses: the rejected,
  the non-rejected and the undecided hypotheses.
  A simulation study motivated by actual biological
  data shows that \texttt{MMCTest} is usable in practice and that, despite
  the additional guarantee, it can be computationally more efficient
  than other methods.
\end{abstract}

\textit{Key words:}
Benjamini Hochberg, Bonferroni correction,
bootstrap, false discovery rate, multiple comparisons,
resampling, sequential algorithm

\section{Introduction}
\label{section_introduction}
Consider multiple hypotheses to be tested 
by individual tests
using a procedure which corrects for  multiplicity, such as the
\cite{Benjamini1995CFD} procedure or the \cite{Bonferroni1936} correction.
The procedure of \cite{Benjamini1995CFD}
has recently received much attention, resulting in various
generalizations \citep{Finner2012,Farcomeni2009,Farcomeni2007,Meinshausen2006}.

Standard procedures require knowledge of the ideal p-values of all tests.
We consider the case where p-values are not known exactly and can only be
computed by simulation.  For example, this occurs for bootstrap
or permutation tests. We will call such tests Monte Carlo tests.
Safe implementations of individual Monte Carlo tests and computation of their
power have been proposed \citep{gandy09:SeqImplMC,GandyRubinDelanchy2013}.
Recent studies involving Monte Carlo tests use a
variety of data sources such as data from a genome data archive
\citep{Pekowska2010}, brain activity data \citep{LageCastellanos2010}
and microarray data \citep{JiaoZhang2010}.
As an example, we consider microarray data of gene expressions for yeast
chemostat cultivations \citep{Knijnenburg2009} in the present article.

This article introduces \texttt{MMCTest}, an algorithm to implement
the multiplicity correction for multiple Monte Carlo tests.
The algorithm gives, with pre-specified probability, the same
classification (rejected and non-rejected hypotheses) as the
classification based on the ideal p-values. For
permutation tests, the ideal p-values can in principle be
obtained by running through all permutations. For bootstrap tests, the
ideal p-value is the probability that a bootstrapped test statistic is
at least as extreme as the observed test statistic.

The motivation for trying to achieve the same classification as the
one obtained with the ideal p-values is mainly repeatability and
objectivity of the results, which \cite{Gleser1996} called
\textit{first law of applied statistics}: ``Two individuals using the
same statistical method on the same data should arrive at the same
conclusion.''  Our algorithm achieves this up to a guaranteed
pre-specified error probability. Another reason for comparing to the
ideal p-values is that all the
theoretical results of the multiple testing procedure based on the
ideal p-values still hold (again up to the guaranteed error
probability).

Our proposed algorithm is sequential: it starts with all hypotheses being
unclassified and then takes samples and classifies
hypotheses until all but a certain
number of hypotheses have been classified or until a certain
effort is reached.
The proposed algorithm can be stopped earlier while having
the same guarantee on the probability of misclassifications.
When stopped before all hypotheses have been
classified, the algorithm returns three sets: the rejected, the
non-rejected and the not yet classified hypotheses.

A literature review detailing existing approaches
used to classify multiple hypotheses without knowledge of the ideal p-values
is presented in Section \ref{section_literature}.

The basic \texttt{MMCTest} algorithm is described in Section
\ref{section_algorithm}. Moreover, Section \ref{section_algorithm}
states conditions which 
bound the probability of
classification errors and which guarantee the convergence
of the testing result of \texttt{MMCTest} to the classification  based on the ideal p-values.
The multiple testing procedure of \cite{Benjamini1995CFD}
and the \cite{Bonferroni1936} correction
satisfy these conditions (Appendix \ref{appendix_properties}).

In Section \ref{section_empirical_comparison},
we first present an application of \texttt{MMCTest}
motivated by real biological data, given by a microarray dataset of
gene expressions for yeast chemostat cultivations
\citep{Knijnenburg2009}.
Afterwards, we conduct simulation studies motivated by this real data to
compare the performance of a naive approach and of \texttt{MCFDR} to
\texttt{MMCTest}. Furthermore, we investigate the dependance of \texttt{MMCTest}
on certain parameters.

We conclude with a discussion in Section \ref{section_discussion}.
All proofs can be found in the Appendix.
The \texttt{MMCTest} algorithm is implemented in an R-package
(\texttt{simctest}, available on CRAN, The Comprehensive R Archive Network).
The \Supp{}
includes an evaluation of a second dataset and
further simulation studies with different testing thresholds.

\section{Literature review}
\label{section_literature}
Several  algorithms  in the literature are related
to our approach in that they aim to stop drawing
samples for certain hypotheses. Some of them aim to give guarantees
on their result, but these guarantees are usually much weaker than our guarantee.

Superficially, our algorithm is close to the algorithm proposed
by \cite{GuoPedadda2008}. Both algorithms maintain confidence
intervals for the p-values corresponding to each hypothesis, and
stop generating samples for hypotheses for which a decision can
be reached. For this both algorithms rely on  the monotonicity
property of the \cite{Benjamini1995CFD} procedure \citep{TamhaneLiu2008}.

However, there are crucial differences. These mainly come from
the different aim of the algorithms: \cite{GuoPedadda2008} aim to
reduce the effort compared to the naive approach with a fixed
number of samples per hypothesis. We aim to give the same
classification as the classification using the ideal p-values.  As
a consequence, their algorithm imposes an upper bound on the
number of samples generated per hypothesis, whereas our algorithm
is open-ended.
Their algorithm does not aim to ensure repeatability,
whereas we aim to do so.
To be able to do this we judiciously control the joint
coverage probability of the intervals.

To be specific, the main results in
\citet[Proposition 1, Theorem 1]{GuoPedadda2008}
are related to Lemma \ref{lemma_monotonicity_sets}
in the present article, with the difference that \cite{GuoPedadda2008}
compare the classification to the naive approach with a fixed number
of samples, whereas we compare the classification to the one based on
the ideal p-values.  Furthermore, the article of \cite{GuoPedadda2008}
has no equivalence to our Theorem \ref{theorem_convergence} in which
we prove that the classification returned by our algorithm converges
to the one based on the ideal p-values.

\cite{JiangSalzman2012} present
an early stopping procedure with a bound on
its computational savings.
As \cite{GuoPedadda2008}, \cite{JiangSalzman2012}
aim at designing a procedure which gives the same result as the naive approach 
with a fixed number of samples. 
Moreover, the authors show that
their procedure only controls the False Discovery Rate (FDR)
\citep{Benjamini1995CFD} up to an error term.

The ad-hoc method of \cite{Wieringen2008} stops
generating samples for hypotheses 
for which a  lower confidence level exceeds a pre-specified
threshold (leading to a non-rejection).
No early stopping for rejections is proposed.
Being an ad-hoc method for a specific application,
no explicit theoretical results are given.

The algorithm \texttt{MCFDR} of \cite{Sandve2011} is a modification of
the algorithm of \cite{BesagClifford1991}, which, for a single
hypothesis, stops drawing further samples when a fixed number of
exceedances has been observed.  The main idea of \texttt{MCFDR} is to
use the criterion of \cite{BesagClifford1991} to obtain quick
non-rejections and to stop the entire algorithm once all remaining
hypotheses are rejected based on their current estimated 
p-values.  Although \texttt{MCFDR} gives quick results, the algorithm
does not give any guarantees on how its result relates to the
result obtained with the ideal p-values.

In contrast to the above approaches, the focus of our algorithm lies
on the computation of the same classification as the one obtained if
the ideal p-value of each test had been available.

There are other related methods which are not necessarily designed to
apply the \cite{Benjamini1995CFD} procedure
or the \cite{Bonferroni1936} correction: The method proposed by
\cite{Knijnenburg2009} uses ordinary permutation p-values if
sufficiently many exceedances can be observed; otherwise, the authors
approximate the p-values using a fitted extreme value distribution.
The aim is to efficiently compute an estimate of all p-values, without
giving any theoretical guarantees.  Several specialized
resampling-based testing procedures for various sampling methods and
various statistics can be found in \cite{WestfallYoung1993}.  All
above methods do not try to take the (unknown) dependance between the
test statistics into account.  Using permutation methods this can be
attempted \citep{Meinshausen2006,WestfallTroendle2008}.

\section{Description of the algorithm}
\label{section_algorithm}
\subsection{Basic algorithm}
\label{subsection_basic_algorithm}
Consider testing $m$ null hypotheses $H_{01}, \ldots, H_{0m}$ having
corresponding test statistics $T_1, \ldots, T_m$ and observed values
$t_1,\ldots,t_m$. A large value of
$t_i$ shall indicate evidence against $H_{0i}$.
Moreover, let $p_i^\ast$ denote the ideal p-value
corresponding to the  hypothesis $H_{0i}$.
We assume that $p^\ast=(p_1^\ast,\ldots,p_m^\ast)$ are not available analytically, but
have to be obtained through simulations.

We assume that for every hypothesis
$H_{0i}$, where $i \in \{ 1,\ldots,m \}$,
we can obtain independent samples from the test statistic
$T_i$ under the null hypothesis.
We will denote these by $T_{ij}$,
and the corresponding exceedance indicators  will be denoted
by $X_{ij} = {\bf 1}({T_{ij} \geq t_i})$, $j \in \N$,
where $\bf 1$ is the indicator function.
In the case of a permutation test, computing  $T_{ij}$ involves generating permutations without replacement.

Suppose that $h: [0,1]^m \rightarrow \mathcal{P}(\{ 1,\ldots,m \})$
takes a vector of p-values and returns the set of indices of
hypotheses to be rejected, where $\mathcal{P}$ denotes the power set.
We will call any such function a \emph{multiple testing procedure}.
Ultimately, we are interested in obtaining $\A$, which we
refer to as the ideal set of rejections.

Following \cite{TamhaneLiu2008},
we call a multiple testing procedure $h$
\emph{monotonic} if $h(p) \supseteq h(q)$ $\forall p \leq q$,
where $p, q \in [0,1]^m$, i.e.\ if lower p-values lead to
more rejections.

The following generic algorithm is designed
for monotonic multiple testing procedures.
It iteratively controls the set of hypotheses
for which further samples need to be drawn by
refining confidence intervals for every $p_i^\ast$ through Monte Carlo sampling.
At iteration $n$, the confidence interval for the p-value $p_i^\ast$
is denoted by $I_i^n$.
The upper confidence limit of a confidence interval $I_i^n$ is
denoted by $\max I_i^n$ and the lower confidence limit
is denoted by $\min I_i^n$.

The following  variables and functions control the behavior of the
algorithm.
The variable $\Delta$ controls
how many additional samples are drawn in each iteration. It is
increased geometrically by a constant $a \geq 1$ in each step
of the algorithm, starting at $\Delta_0 \geq 1$.
In the examples of this article we use $a=1.25$ and
$\Delta_0=10$. Two vectors $S,k \in \N_0^m$ keep track of counts.

The function $f(S,k,\Delta)$
computes a confidence interval for the ideal p-value of a hypothesis
based on the number of exceedances $S$ and the number of  samples $k$ drawn
for this hypothesis. The dependance on 
the current value of $\Delta$ is needed to be able to
guarantee a joint coverage probability of all confidence intervals
produced in the algorithm. 
For simplicity, we will assume that $f$ returns  closed confidence intervals.
In Appendix \ref{appendix_clopper_pearson} we give an example
for such an $f$ which computes \cite{ClopperPearson1934} confidence
intervals and uses a spending sequence to guarantee an overall
coverage probability.

The algorithm runs until at most $c \geq 0$
hypotheses are classified or until the total number of
samples drawn reaches a pre-specified  limit $k_{\max}$.
The following pseudo-code uses $c=0$ and $k_{\max}=\infty$,
thereby computing a classification of all hypotheses.

In the remainder of this article, $| \cdot |$ denotes the number of elements
in a finite set and the length of an interval. Moreover,
$\left\Vert \cdot \right\Vert$ denotes the Euclidean norm of a vector.

\begin{algorithm}
\label{algorithm_basic}
$\textnormal{\bf [MMCTest}((X_{ij})_{ij},f,c=0,k_{\max}=\infty,\Delta_0=10,a=1.25)\textnormal{\bf ]}$
\begin{itemize}
  \item[] $n:=0$, $\Delta := \Delta_0$, $\underline{A}_0 := \emptyset$, $\overline{A}_0:=\{ 1,\ldots, m \}$.
  \item[] $I_i^0 := [0,1]$, $S_i:=k_i:=0$ $\forall i = 1,\ldots,m$.
  \item[] While $\left( |\overline{A}_n \setminus \underline{A}_n|>c
  \text{ and } \sum_{i=1}^m k_i \leq k_{\max} \right)$
  \begin{itemize}
    \item[] $n := n+1$.
    \item[] $\Delta := \lfloor a \Delta \rfloor$.
    \item[] For all $i \in \overline{A}_{n-1} \setminus \underline{A}_{n-1}$:
    \begin{itemize}
      \item[] $S_i := S_i + \sum_{j=k_i+1}^{k_i+\Delta} X_{ij}$.
      \item[] $k_i := k_i + \Delta$.
      \item[] $I_i^n := f(S_i,k_i,\Delta) \cap I_i^{n-1}$.
    \end{itemize}
    \item[] For all $i \notin \overline{A}_{n-1} \setminus \underline{A}_{n-1}$: $I_i^n:=I_i^{n-1}$.
    \item[] Set $\underline{A}_n := h( (\max I_i^n)_{i=1,\ldots,m})$, $\overline{A}_n := h( (\min I_i^n)_{i=1,\ldots,m})$.
  \end{itemize}
  \item[] Return $(\underline{A}_n, \overline{A}_n)$.
\end{itemize}
\end{algorithm}

The algorithm works as follows: The number of additional
samples $\Delta$ drawn in every step is increased
geometrically. The total number of samples drawn up to iteration $n$
for a hypothesis $i \in \{1,\ldots,m\}$ is stored in $k_i$
and the total number of observed exceedances is stored in $S_i$.
For all  hypotheses which are still under consideration, i.e.\ those in 
$\overline{A}_{n-1} \setminus \underline{A}_{n-1}$, an additional
batch of $\Delta$ samples is drawn and new confidence intervals are
computed.
The confidence intervals remain unchanged for the other hypotheses.
New classifications are then computed based
on the updated upper and lower confidence limits.

The confidence intervals $I_i^n$ computed
in Algorithm \ref{algorithm_basic} are nested by construction.

\begin{figure}[tb]
  \centering
  \includegraphics[width=\textwidth]{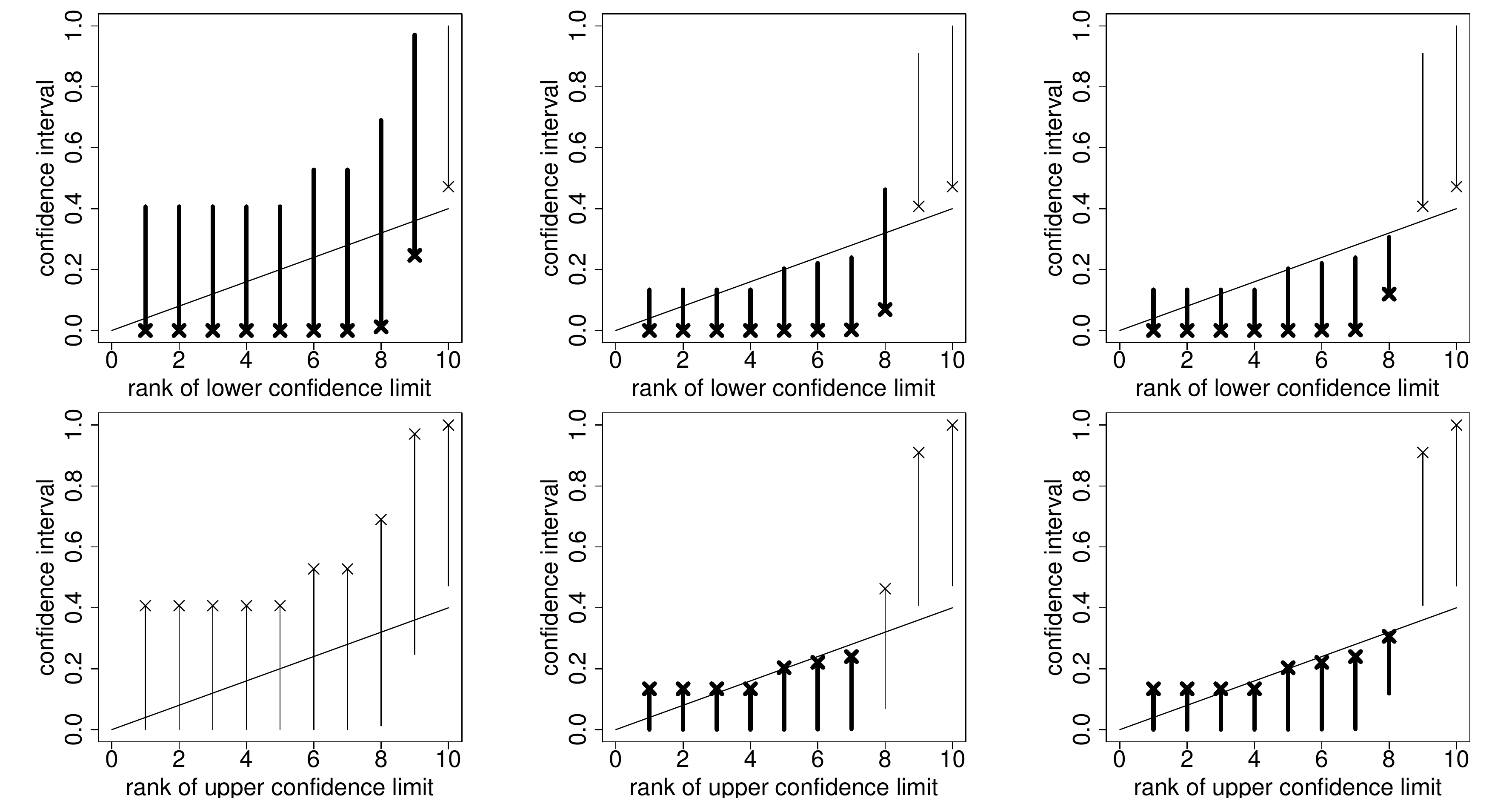}
  \caption{
Example run of \texttt{MMCTest} on $m=10$ hypotheses using  the
Benjamini-Hochberg procedure $h$:  after the second
iteration (left), after a few additional iterations
(center) and after the last iteration (right).
Bold confidence intervals denote elements of $\overline{A}_n$ in the
upper row and elements of $\underline{A}_n$ in the lower row. The lower (upper) confidence limits used to compute
$\overline{A}_n$ ($\underline{A}_n$) are marked with a cross.}
  \label{figure_visualisation}
\end{figure}

\begin{example}
\label{example_visualization} 
An example run of \texttt{MMCTest} (with $m=10$ hypotheses
and $c=0$) is shown in Figure \ref{figure_visualisation}.
We use the \cite{Benjamini1995CFD} FDR controlling procedure
(see Appendix \ref{appendix_benjamini_hochberg})
with threshold $\alpha=0.4$ as the multiple testing function $h$.
The function $f$ given in Appendix \ref{appendix_clopper_pearson} is used
to compute confidence intervals.
Columns show different iterations, the upper row
shows the computation of $\overline{A}_n$, the lower row
shows the computation of $\underline{A}_n$.
The indices contained in $\overline{A}_n$ and $\underline{A}_n$
are visualized with bold confidence intervals.
Additionally, the lower (upper) confidence limits used to compute
$\overline{A}_n$ ($\underline{A}_n$) are marked with a cross.
Only the lower (upper) end of the confidence interval matters
for the computation of $\overline{A}_n$ ($\underline{A}_n$),
thus the hypotheses are ordered by their lower
(upper) confidence limit in the upper (lower) row.
In this example this turns out to be the same ordering.
After the second iteration (left column),
\texttt{MMCTest} has already classified the last hypothesis as being non-rejected
as the lower confidence limit of its p-value lies above the line
connecting the points $(0,0)$ and $(m,\alpha)$ which we call the
Benjamini-Hochberg line.
All other hypotheses are still undecided and thus their confidence intervals
will be refined. After a few additional iterations (middle column), the seven smallest
values can be classified as rejected as the upper confidence limit of the seventh value
is below the line. Likewise, the confidence interval of the ninth value has now been shrunk
to be entirely above the line which classifies this value as non-rejected.
The eighth p-value is still unclassified as its confidence interval overlaps with the line.
After refining the confidence interval further, the algorithm stops in the
situation depicted in the right column with a complete classification
($\overline{A}_n=\underline{A}_n$).
\end{example}

The monotonicity of $h$ implies immediately
that the sequence of sets $\underline{A}_n$
is increasing, that the sequence of sets $\overline{A}_n$
is decreasing and,
on an additional assumption,
that each $\underline{A}_n$ ($\overline{A}_n$) is a subset (superset)
of the ideal set of rejections $\A$.

\begin{lemma}
\label{lemma_monotonicity_sets}
Assume that $h$ is monotonic.
\begin{enumerate}
  \item $(\underline{A}_n)_{n \in \N} \nearrow$ and $(\overline{A}_n)_{n \in \N} \searrow$.
  \item If $p_i^\ast \in I_i^n$ $\forall i,n$, then
$\underline{A}_n \subseteq \A \subseteq \overline{A}_n$ $\forall n \in \N$.
\end{enumerate}
\end{lemma}

\subsection{Conditions and main results}
\label{subsection_conditions}
In this section we show that under certain conditions the
classification of \texttt{MMCTest} is correct with high probability,
meaning that all classifications are identical to the
classifications based on the ideal p-values. Furthermore, we show
that all hypotheses will be classified.

The first condition pertains to  the multiple testing procedure $h$.
Besides asking for monotonicity, it ensures
that lowering the p-value of a rejected hypothesis or
increasing the p-value of a non-rejected hypothesis
does not change the result of $h$.

\begin{condition}
\label{condition_invariance}
\begin{enumerate}
  \item $h$ is monotonic.
  \item Let $p,q \in [0,1]^m$.
If $q_i \leq p_i$ $\forall i \in h(p)$ and $q_i \geq p_i$ $\forall i \notin h(p)$,
then $h(p) = h(q)$.
\end{enumerate}
\end{condition}

The second condition requires the function $f$
to produce confidence intervals whose length goes uniformly to 0 as
more samples are drawn.

\begin{condition}
\label{condition_intervals}
$|f(S,k,\Delta)|$ converges uniformly to $0$ as $k \to \infty$, i.e.\
$\forall \epsilon>0$ $\exists k_0 \in \N$ such that $\forall k \geq k_0$,
$\forall S \in \{ 0,\ldots,k \}$ and $\forall \Delta \in \N$,
we have $|f(S,k,\Delta)| < \epsilon$.
\end{condition}

The main theorem guaranteeing convergence is as follows:

\begin{theorem}
\label{theorem_convergence}
Suppose Conditions \ref{condition_invariance} and
\ref{condition_intervals} hold and suppose
that there exists $\delta>0$ such that $p \in [0,1]^m$ and
$\left\Vert p-p^\ast \right\Vert <\delta$ imply $h(p)=h(p^\ast)$.
Then, on the event $\{ p_i^\ast \in I^n_i ~\forall i,n \}$,
both sequences $(\underline{A}_n)_{n \in \N}$ and
$(\overline{A}_n)_{n \in \N}$ converge to $\A$,
i.e.\ there exists $n_0 \in \N$ such that
$\underline{A}_n = \overline{A}_n = \A$ $\forall n \geq n_0$.
\end{theorem}

The condition on $p^\ast$ in Theorem \ref{theorem_convergence}
ensures that $p^\ast$ has a neighborhood on which $h$ is constant.

The FDR controlling procedure of \cite{Benjamini1995CFD}
(see Appendix \ref{appendix_benjamini_hochberg})
and the \cite{Bonferroni1936} correction
(see Appendix \ref{appendix_bonferroni_correction})
both satisfy Condition \ref{condition_invariance}
(see Corollary \ref{corollary_benjamini_hochberg} and
Corollary \ref{corollary_Bonferroni} in Appendix \ref{appendix_properties})
and the condition on $p^\ast$ in Theorem \ref{theorem_convergence}
(see Lemma \ref{lemma_invariance_delta} and
Lemma \ref{lemma_Bonferroni_invariance_delta}
in Appendix \ref{appendix_properties}) for almost all $p^\ast$.

The following third condition ensures that the confidence
intervals computed by the function $f$ in Algorithm \ref{algorithm_basic}
have a guaranteed  joint coverage probability.
The choice of $f$ given in Appendix \ref{appendix_clopper_pearson} satisfies
Condition \ref{condition_intervals} and Condition \ref{condition_probability}
(see Lemma \ref{lemma_clopper_pearson_satisfy_condition}
in Appendix \ref{appendix_clopper_pearson}).

\begin{condition}
\label{condition_probability}
For a given $\epsilon>0$, the function $f$ computes
confidence intervals $I_i^n$ in such a way that
$\PP(p_i^\ast \in I_i^n ~ \forall i,n) \geq 1-\epsilon$.
\end{condition}

The main theorem and Condition \ref{condition_probability} together
immediately give a bound on the probability of misclassifications.

\begin{corollary}
\label{corollary_convergence}
Under the conditions of Theorem \ref{theorem_convergence} and
under Condition \ref{condition_probability},
$$\PP(\exists n_0: \underline{A}_n = \A = \overline{A}_n ~\forall n \geq n_0) \geq 1-\epsilon,$$
i.e.\ the probability that all classifications are correct is at least $1-\epsilon$.
\end{corollary}

\section{Simulation studies}
\label{section_empirical_comparison}
The first aim of this section is to demonstrate that \texttt{MMCTest}
can be used to classify thousands of hypotheses commonly encountered in real data studies
(Section \ref{subsection_real_application}).
Moreover, this section shows that when matching the effort,
\texttt{MMCTest} computes classifications
containing a number of
unclassified hypotheses which is comparable
to the number of misclassifications incurred by
current approaches like the naive method or the \texttt{MCFDR}
algorithm -- even though \texttt{MMCTest}
is able to guarantee the correctness of all its classified hypotheses
while for the two other methods,
misclassified hypotheses typically remain
unidentified in the testing result
(Section \ref{subsection_comparison_naive}
and Section \ref{subsection_comparison_mcfdr}).
An ad-hoc variant of \texttt{MMCTest}
computing a complete classification yields less misclassifications and random
classifications than the other methods, demonstrating
that \texttt{MMCTest} is the superior method for practical applications.
The aim of the last two sections is to empirically investigate the
dependance of \texttt{MMCTest} on certain parameters.
Section \ref{subsection_effort_dependence_number_hypotheses} studies the
dependance of the computational effort of \texttt{MMCTest} on the
number of hypotheses $m$.
We conclude by empirically assessing the runtime of \texttt{MMCTest} in
Section \ref{subsection_effort_dependence_unclassified_hypotheses}, demonstrating that whilst
a complete classification can be computationally very expensive, 
most hypotheses can be classified with a reasonable effort.

\subsection{The set-up}
\label{subsection_setup}
The following parameters were used throughout Section \ref{section_empirical_comparison}.
The batch size $\Delta$ in Algorithm \ref{algorithm_basic}
is increased by $a=1.25$ in every iteration,
starting with $\Delta_0=10$.
Confidence intervals are computed
using the function $f$
with \cite{ClopperPearson1934} confidence intervals
and parameters $\epsilon=0.01$ and $r=10000$
(see Appendix \ref{appendix_clopper_pearson}).
The Benjamini-Hochberg procedure
(at threshold $\alpha=0.1$) as defined in Appendix \ref{appendix_benjamini_hochberg}
always serves as multiple testing procedure.

We measure the effort of any algorithm in terms of $N$,
the total number of samples drawn during a run.

We use a yeast chemostat cultivation dataset of
\cite{Knijnenburg2009}. This dataset consists of $170$ microarrays of
yeast cultivations. The first $80$ microarrays correspond
to yeast which was grown aerobically, the second $90$
microarrays correspond to yeast which was grown anaerobically.
Every microarray reacts to $9335$ genes, thus giving rise to
$9335$ null hypotheses (no effect of the gene onto the response).

To speed up the computation of the simulation
studies in this and the following sections
as well as to have an underlying ``truth'' for
the \cite{Knijnenburg2009} dataset,
we estimated each of the $m=9335$ p-values once
by generating $10^6$ permutations per hypothesis as outlined in
the \Supp{}.
Such a number of permutations is far more
than what would commonly be used in practice.
We then define these approximated p-values to be the ideal
p-values $p_1^\ast,\ldots,p_m^\ast$
we are interested in, although they do not necessarily have to be
equal to the p-values underlying each hypothesis.
A plot of the ideal p-values $p_1^\ast,\ldots,p_m^\ast$
is given in the \Supp{}.

In the following sections,
we draw Bernoulli samples
with success probabilities $p_1^\ast,\ldots,p_m^\ast$
instead of generating actual permutations.
The classification obtained by applying
the Benjamini-Hochberg procedure directly to
the ideal p-values $p_1^\ast,\ldots,p_m^\ast$
is used to compute misclassifications.

The \Supp{} contains similar simulations as the ones
which are about to follow for two additional testing thresholds $\alpha$ (1\% and 5\%). 
The behavior of \texttt{MMCTest} is qualitatively similar.
Furthermore, the \Supp{} contains 
another comparison of \texttt{MMCTest}
to the naive method and to \texttt{MCFDR} on a simulated dataset with a larger
proportion of true null hypotheses than the one of
the dataset of \cite{Knijnenburg2009}, broadly
confirming the qualitative results of Sections
\ref{subsection_comparison_naive} and \ref{subsection_comparison_mcfdr}.

\subsection{Application to Real Data}
\label{subsection_real_application}
\texttt{MMCTest} is applied once
to the ideal p-values as described in Section \ref{subsection_setup}.

After having drawn $24.5 \cdot 10^6$ samples all but $100$ hypotheses are classified.
This corresponds to only around $2600$ samples per hypothesis,
thus making a classification with such a precision fairly easy to compute.
Drawing roughly  the same number of samples again
(a total number of $49.7 \cdot 10^6$ samples)
classifies all but $50$ hypotheses.

\texttt{MMCTest} can be stopped whenever the user's desired number of classifications is achieved.
All but $20$ hypotheses are classified after  $159 \cdot 10^6$ samples and all
but $10$ hypotheses after  $255 \cdot 10^6$ samples. A classification
of all but $5$ hypotheses is obtained after having drawn a total number
of $12 \cdot 10^9$ samples. This is, of course, extremely computationally intensive.
The total number of samples drawn for a classification
of all but $5$ hypotheses corresponds to roughly $1.3 \cdot 10^6$ samples per hypothesis.

A comparison to the classification result obtained by
applying the \cite{Benjamini1995CFD} procedure to the ideal p-values
shows that in all the classifications previously reported,
none of the decided hypotheses was wrongly classified.

\subsection{Comparison to the naive method}
\label{subsection_comparison_naive}
\begin{table}[tb]
\centering
\caption{Comparison of the naive method to \texttt{MMCTest}}\label{table_comparison_naive}
\begin{tabular}{rrrr||r|rr}
\multicolumn{4}{c||}{naive method} & \multicolumn{3}{c}{MMCTest}\\
&&&& guaranteed classification & \multicolumn{2}{c}{forced classification}\\
\hline
{\it s} & {\it mis} & {\it rc} & {\it N} & {\it unclassified hypotheses} & {\it mis} & {\it rc}\\
\hline
100	&	238	&	680	&	933500	&	7677	&	236	&	686	\\
1000	&	69	&	230	&	9335000	&	317	&	20	&	62	\\
10000	&	21	&	65	&	93350000&	32	&	3	&	6	\\
\hline
\end{tabular}
\begin{minipage}{\textwidth}
$s$: number of samples used by the naive method for each hypothesis;
\textit{mis}: average number of misclassifications;
\textit{rc}: number of randomly classified hypotheses;
$N$: average total number of samples.
\end{minipage}
\end{table}

We compare \texttt{MMCTest} to the sampling scheme which draws a
constant number of samples $s$ for each hypothesis.
It then estimates each p-value via its proportion of exceedances
(a formula for this estimate is given in the \Supp{})
and computes a classification by
applying the multiplicity correction to the estimates,
thus treating the estimated p-values as if they were ideal p-values.
We will call this the naive method.
The naive method is widely used in connection with the
False Discovery Rate approach to
evaluate real biological data
\citep{Cohen2012,Gusenleitner2012,Nusinow2012,Rahmatallah2012}.

The results presented in this and the following section
are based on $10000$ runs.
In each run, we draw Bernoulli samples for the naive method
and for \texttt{MMCTest} as described in
Section \ref{subsection_setup}.
The sampling standard deviation of averages
is less than the least significant digit we report in tables.

Table \ref{table_comparison_naive} shows the simulation results.
The second column displays the average number of misclassifications 
for the naive method. A considerable
number of misclassifications occurs;
even when using $s=10000$ samples to estimate each p-value
about $21$ misclassifications still occur
on average for the naive method.

The third column in Table \ref{table_comparison_naive} shows an
alternative criterion, the number of randomly classified hypotheses
(\textit{rc}), which we define as follows. Let $f^r_i$ be the frequency
of rejection of hypothesis $H_{0i}$ in the 10000 runs.
A hypothesis $H_{0i}$ is considered to be randomly classified
if $\min(f^r_i,1-f^r_i)$ is strictly larger than 0.1.

The number of randomly classified hypotheses is substantially larger than the 
average number of misclassifications. This  demonstrates that for a substantial number of hypotheses,
the decision reported is mainly determined by the Monte Carlo simulation (and not by the observed data).

The total number of samples $N$ drawn during each run of the naive method
is given in the fourth column of Table \ref{table_comparison_naive}.

\texttt{MMCTest} is run on the ideal p-values
(see Section \ref{subsection_setup})
using at most the total number of samples the naive method had used.
The fifth column in Table \ref{table_comparison_naive} shows the
average number of remaining unclassified hypotheses upon termination.

The average number of unclassified hypotheses of
\texttt{MMCTest} is larger than the number of misclassifications
of the naive method. However, \texttt{MMCTest} gives a
result which is proven to be reliable with pre-specified probability
in contrast to the one computed by the naive method.
For large values of $s$, \texttt{MMCTest} yields
average numbers of unclassified hypotheses which almost
equal the number of misclassifications observed for the
naive method even though
\texttt{MMCTest} guarantees the correctness
of its classified hypotheses while
the misclassifications in the testing result
of the naive method typically remain unidentified.

As shown in Table \ref{table_comparison_naive},
at a high precision,
\texttt{MMCTest} yields
less unclassified hypotheses than the naive method
yields randomly classified hypotheses.
This indicates that \texttt{MMCTest} gets competitive
for a realistic precision and starts overtaking the naive method for
multiple testing settings which are evaluated at high precision.

\texttt{MMCTest} is stopped on reaching the number of
samples used by the naive method. Nevertheless, all the theoretical
guarantees stated in Section \ref{section_algorithm} are still valid,
but not all hypotheses are being classified.
A complete classification in an ad-hoc fashion can
be obtained by applying the multiple testing procedure $h$
to the p-value estimates $\hat{p}_i = (S_i+1)/(k_i+1)$
after stopping ($S_i$ and $k_i$ are as in Algorithm
\ref{algorithm_basic}). The theoretical guarantees
of Section \ref{section_algorithm} are not valid any more for
the ad-hoc procedure.

The two last columns of Table \ref{table_comparison_naive}
show the average number of misclassifications and the number of randomly
classified hypotheses for the ad-hoc procedure which forces a
complete classification upon termination.
With this simple modification, \texttt{MMCTest} yields considerably
lower numbers of misclassifications and randomly classified hypotheses
for a high precision than the naive method.

The forced classification should only be used if a complete
classification is needed within a limited effort. In all other
cases, whenever the algorithm is stopped, we recommend using
the partioning of the hypotheses into rejected, non-rejected and
not classified hypotheses as testing result of the algorithm.

\subsection{Comparison to \texttt{MCFDR}}
\label{subsection_comparison_mcfdr}

\begin{table}[tb]
\centering
\caption{Comparison of \texttt{MCFDR} to \texttt{MMCTest}}\label{table_comparison_mcfdr}
\begin{tabular}{rrrr||r|rr}
\multicolumn{4}{c||}{MCFDR} & \multicolumn{3}{c}{MMCTest}\\
&&&& guaranteed classification & \multicolumn{2}{c}{forced classification}\\
\hline
{\it $u$} & {\it mis} & {\it rc} & {\it N} & {\it unclassified hypotheses} & {\it mis} & {\it rc}\\
\hline
10	&	172	&	524	&	$1.3\cdot 10^6$	&	7547	&	208	&	592	\\
20	&	123	&	389	&	$2.2\cdot 10^6$	&	7268	&	130	&	398	\\
50	&	80	&	267	&	$5.3\cdot 10^6$	&	763	&	50	&	168	\\
100	&	57	&	186	&	$10.5\cdot 10^6$	&	288	&	18	&	54	\\
200	&	40	&	128	&	$21.0\cdot 10^6$	&	132	&	9	&	29	\\
500	&	25	&	72	&	$52.5\cdot 10^6$	&	49	&	4	&	8	\\
1000	&	18	&	54	&	$104.9\cdot 10^6$	&	30	&	3	&	6	\\
\hline
\end{tabular}
\begin{minipage}{\textwidth}
$u$: number of test statistics exceeding the reference statistic
(tuning parameter of \texttt{MCFDR});
\textit{mis}: average number of misclassifications;
\textit{rc}: number of randomly classified hypotheses;
$N$: average total number of samples.
\end{minipage}
\end{table}

We now focus on a comparison of \texttt{MMCTest} to
\texttt{MCFDR} of \cite{Sandve2011},
given in Table \ref{table_comparison_mcfdr}.
\texttt{MCFDR} is run first on
the ideal p-values (see Section \ref{subsection_setup}) already used
for the comparison of \texttt{MMCTest} to the naive method
and \texttt{MMCTest} is then applied with matched effort.
The \texttt{MCFDR} algorithm has one tuning parameter:
the number $u$ of test statistics exceeding the reference
statistic before stopping (this number was called $h$ in \cite{Sandve2011}).
In \cite{Sandve2011} the authors
recommend using $u=20$, but we will also consider larger values.

In its original statement in \cite{Sandve2011}, the
\texttt{MCFDR} algorithm uses
a modification of the Benjamini-Hochberg procedure of
\cite{PoundsCheng2006} which uses an
estimate $\hat{\pi}_0(p)$ of the proportion of true null hypotheses.
\texttt{MCFDR} can also
be used together with the standard Benjamini-Hochberg procedure by
setting $\hat{\pi}_0(p)$ 
to one. The following results have
been computed using the standard Benjamini-Hochberg procedure
(as defined in Section \ref{appendix_benjamini_hochberg})
for both \texttt{MCFDR} and \texttt{MMCTest}.

The first columns of Table \ref{table_comparison_mcfdr} show
the average number of misclassifications
\textit{mis} and the number of randomly classified
hypotheses \textit{rc} for \texttt{MCFDR} for various values of $u$.
Similar to Table \ref{table_comparison_naive},
the number of randomly classified hypotheses
occurring for \texttt{MCFDR}
is generally larger than the average number
of misclassifications.

When using \texttt{MMCTest}, the number of unclassified hypotheses is
generally larger than the number of misclassifications for \texttt{MCFDR}.
The advantage of using \texttt{MMCTest} is, as before,
the guaranteed classification.

\texttt{MMCTest} becomes more competitive for higher precisions.
For large values of $u$, the \texttt{MMCTest} algorithm
classifies all hypotheses with
confidence up to a number which almost equals the number of
misclassifications of \texttt{MCFDR}, and which is less
than the number of randomly classified hypotheses of \texttt{MCFDR}.

The forced classification in \texttt{MMCTest}
(the last two columns in Table \ref{table_comparison_mcfdr})
yields a considerably better classification than \texttt{MCFDR} for high precisions,
both in terms of misclassifications and
randomly classified hypotheses.

\subsection{Dependance of the effort on the number of hypotheses}
\label{subsection_effort_dependence_number_hypotheses}

\begin{figure}[tb]
  \centering
  \includegraphics[width=0.5\textwidth]{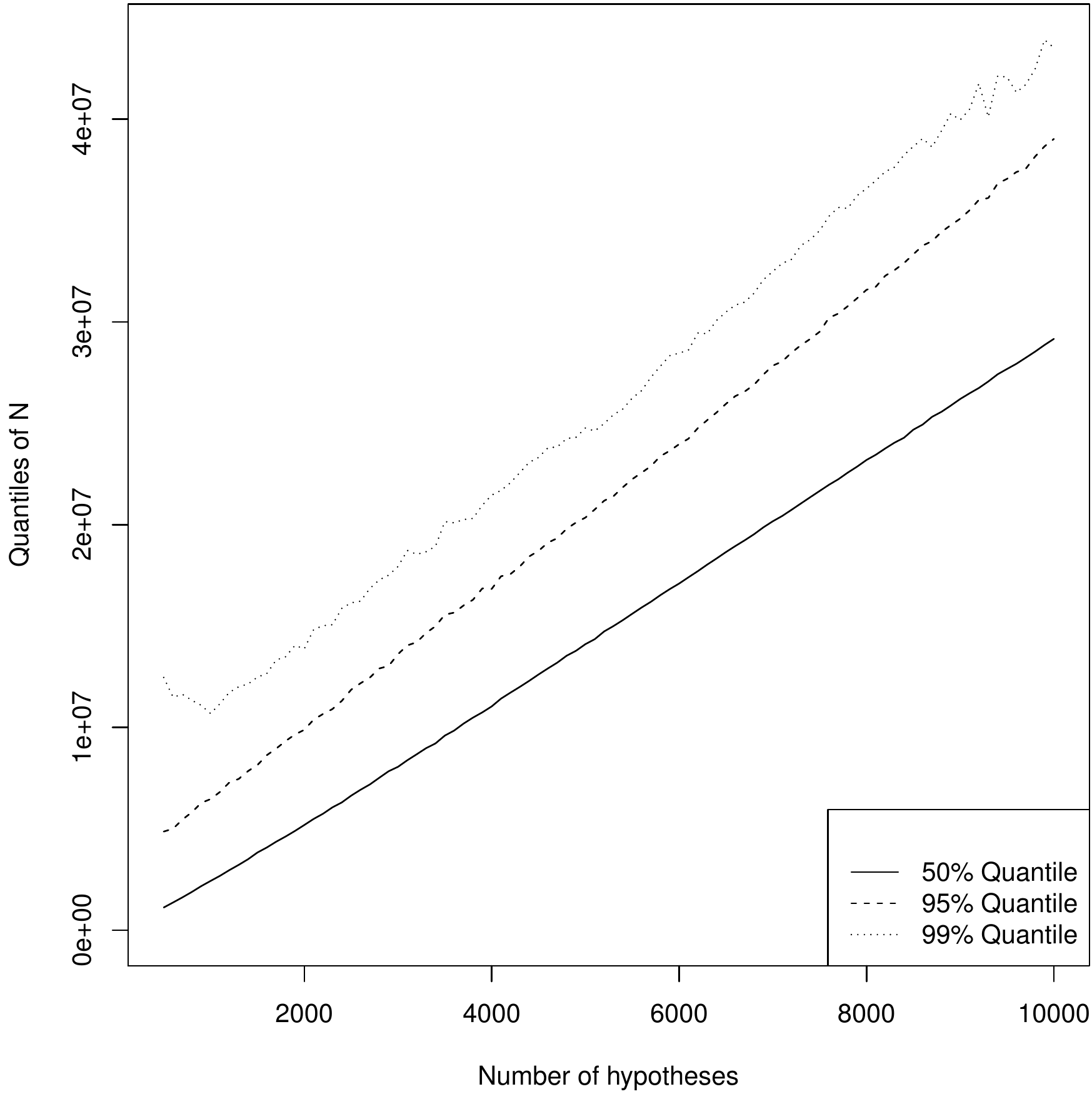}
  \caption{$50\%$-, $95\%$- and $99\%$-quantiles
of the effort $N$ against the number of hypotheses $m$.
Quantiles are computed based on $10000$ runs
classifying all but $c=0.01m$ hypotheses.
P-values for various values of $m$ are obtained by resampling with
replacement from the ideal p-values
(see Section \ref{subsection_setup}).}
  \label{figure_dependence_number_hypotheses}
\end{figure}

How does the number of samples $N$ depend on the number of hypotheses?

Figure \ref{figure_dependence_number_hypotheses} shows
$50\%$-, $95\%$- and $99\%$-quantiles of the effort $N$
for a classification of $m$ hypotheses, where
$m$ ranges from $500$ to $10000$ in steps of $100$.
Quantiles are computed based on $10000$ repetitions.
For each value of $m$ and each repetition,
a new p-value distribution is obtained by resampling with
replacement from the ideal p-values $p_1^\ast,\ldots,p_m^\ast$
(see Section \ref{subsection_setup}).
\texttt{MMCTest} is then
run on the new distribution obtained in this way
until all but $c=0.01m$ hypotheses are classified.

Figure \ref{figure_dependence_number_hypotheses} indicates that
the effort $N$ for a classification of all but $c=0.01m$ hypotheses
increases linearly in $m$.

\subsection{Dependance of the effort on the number of unclassified hypotheses}
\label{subsection_effort_dependence_unclassified_hypotheses}

\begin{figure}[tb]
  \centering
  \includegraphics[width=0.5\textwidth]{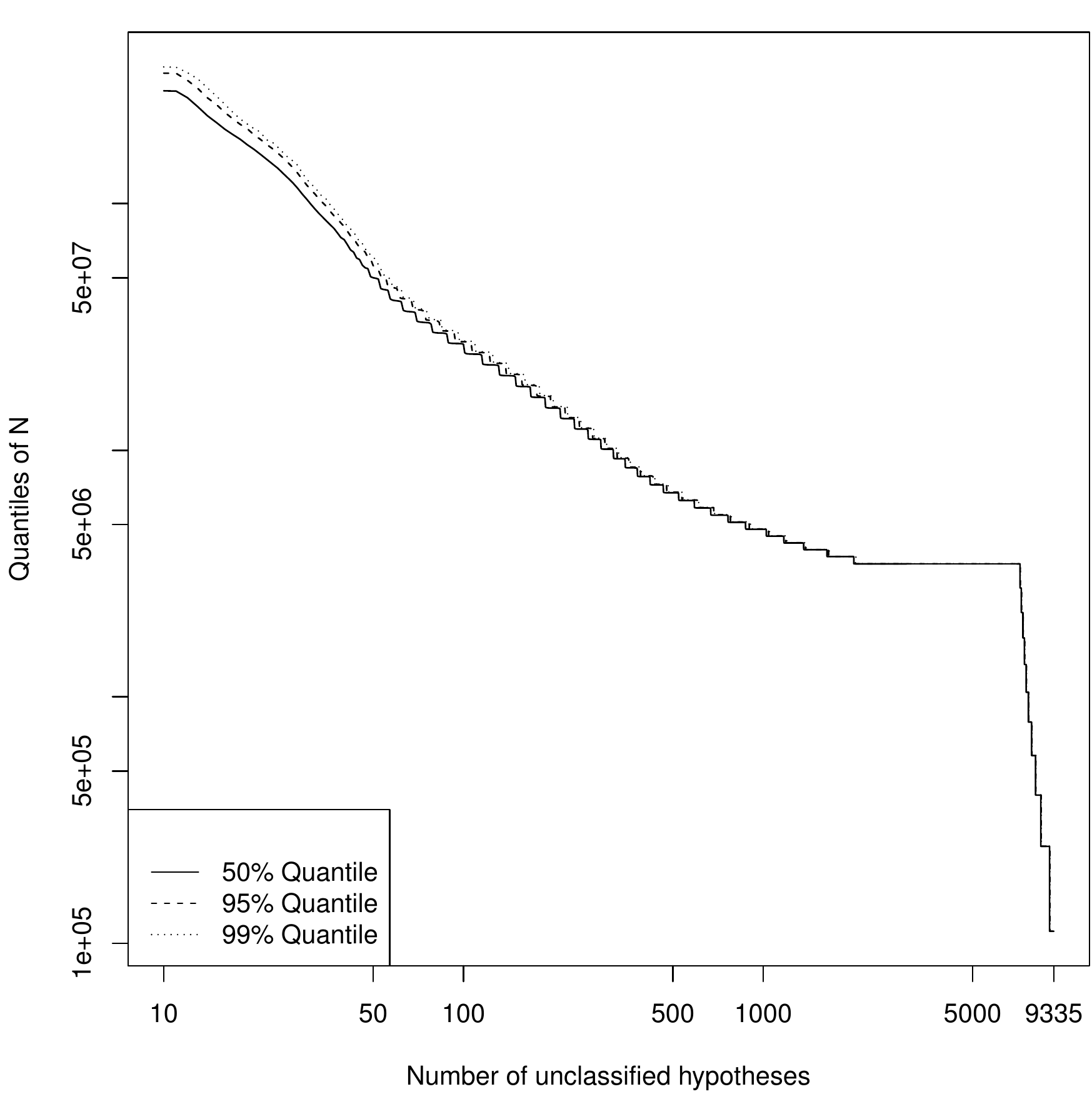}
  \caption{Effort $N$ against number of unclassified hypotheses.
$50\%$-, $95\%$- and $99\%$-quantiles of $N$
based on $1000$ simulations. Log-scale on both axes.
The classification becomes more complete
as the quantile curves approach the left hand side.}
  \label{figure_complexity}
\end{figure}

Figure \ref{figure_complexity} shows the dependance of the
effort $N$ on the number of unclassified hypotheses for a
fixed number of hypotheses $m$.
The right hand side of Figure \ref{figure_complexity} corresponds
to the situation of all hypotheses being unclassified.
The classification becomes more complete
as the quantile curves approach the left hand side.

To generate this figure, \texttt{MMCTest} is applied $1000$ times
in the following way to the ideal p-values
$p_1^\ast,\ldots,p_m^\ast$
($m=9335$, see Section \ref{subsection_setup}):
The current size $c$ of the
set $\overline{A}_n \setminus \underline{A}_n$
and the current total number of samples $N_c$
are recorded in each iteration $n$.
If several p-values are classified together
in an iteration, some $c$ do not have a corresponding $N_c$. To be
conservative, a missing value  $N_{c}$ is set to $N_{c'}$ for the
largest $c' < c$ for which $N_{c'}$ is not missing.
Each time the algorithm is run until all but $c=10$
hypotheses are classified.

The effort is reasonable for classifying all but a few
hypotheses. Classifying the last few hypotheses seems to be
computationally intensive.

The steps in Figure \ref{figure_complexity} are caused by several
hypotheses with p-values far off the Benjamini-Hochberg
line being classified together. This effect also occurs
in Figure 2 of the \Supp{} which shows that
at a certain iteration $n$, several hypotheses are classified together,
thereby causing a sudden increase in the size of the set $\underline{A}_n$.

\section{Discussion}
\label{section_discussion}
We presented an open-ended sequential algorithm designed to
implement multiplicity corrections for
multiple Monte Carlo tests in the setting
where the ideal p-values are unknown and
can only be approximated through simulation.
In order to ensure repeatability and objectivity for
Monte Carlo based multiple testing, we aim
to compute the same classification as the one based on the ideal p-values.

The main feature of  \texttt{MMCTest}  is that its output is guaranteed to
be correct with a pre-specified probability, meaning that
all its classifications are identical to the classifications based on the
ideal p-values.

Our simulation study shows that a complete classification can be
computationally  expensive, but that  most hypotheses can be classified using a
reasonable effort.
For a realistic precision, \texttt{MMCTest} draws level
with the performance of current methods which unlike \texttt{MMCTest} do
not give a guarantee on their classifications being correct, such as the naive approach or
the \texttt{MCFDR} algorithm. An ad-hoc variant of \texttt{MMCTest}
outperforms the naive method and \texttt{MCFDR} both in terms of
misclassifications and randomly classified hypotheses.
Tuning the parameter $r$ of the spending sequence,
spending all the remaining error probability or 
matching the effort exactly in the last iteration
leaves scope for further research.
A detailed theoretical analysis of the computational effort of the proposed algorithm
is outside the scope of this article.

This article identified conditions which guarantee the bounded
risk of classification errors and the convergence of the algorithm's
output to the classification computed with the ideal p-values. By
verifying these conditions we showed that the \texttt{MMCTest}
algorithm works for the \cite{Benjamini1995CFD} procedure
as well as for the \cite{Bonferroni1936} correction.
We conjecture that our algorithm also works for other FDR controlling
procedures and for procedures controlling FDR-related criteria (e.g.\
the False Non-Discovery Rate FNR).

\section*{\Supp{}}
Additional information for this article is available online
including further details on Section \ref{section_empirical_comparison}.
In particular, it contains details on the computation of the test statistic,
a similar set of simulation studies
for two additional testing thresholds $\alpha=0.01$
and $\alpha=0.05$ as well as another comparison of \texttt{MMCTest}
to the naive method and to \texttt{MCFDR} on a simulated dataset with a larger
proportion of true null hypotheses than the one used in the present article.

\appendix

\section{Clopper-Pearson confidence intervals}
\label{appendix_clopper_pearson}
Our particular choice of the function $f$ used
in Example \ref{example_visualization} and the empirical studies
in Section \ref{section_empirical_comparison}
computes ``exact'' \cite{ClopperPearson1934} confidence intervals.
We choose $f$ in such a way as to guarantee a joint coverage probability
of $1-\epsilon$ for all confidence intervals over all iterations,
where the overall error probability $\epsilon$ is chosen by the user.

A sequence $\left( \eta_k \right)_{k \in \N_0}$
satisfying $\eta_0=0$ and $\eta_k \rightarrow \epsilon$ as $k \rightarrow \infty$
is used to control how $\epsilon$ is spent over the iterations
of the algorithm. We will call $\left( \eta_k \right)$
\texttt{spending sequence}.
Throughout the article we use $\eta_k := \frac{k}{k+r} \epsilon$
for some constants $r>0$ and $\epsilon>0$
(all the parameters used for the simulation studies are
given in Section \ref{subsection_setup}).

We then define $f(S,k,\Delta)$ to be
the \cite{ClopperPearson1934} confidence interval
based on $S$ and $k$ (see Algorithm \ref{algorithm_basic})
with a coverage probability of $1-(\eta_k-\eta_{k-\Delta})/m$.
Precisely,
$$f(S,k,\Delta) := \begin{cases}
[ 1 - q^\text{Beta}_{k-S, S+1}(\rho_{k}),
1 - q^\text{Beta}_{k+1-S, S}(1 - \rho_{k}) ] & 0<S<k,\\
[ 0, 1-\rho_{k}^{1/k} ] & S=0,\\
[ \rho_{k}^{1/k}, 1 ] & S=k,
\end{cases}$$
where $\rho_{k}=(\eta_{k}-\eta_{k-\Delta})/(2m)$.
The quantiles $q^\text{Beta}_{\alpha, \beta}(\epsilon)$ of the
$\text{Beta}(\alpha,\beta)$ distribution being used are defined by
$\PP(Z \leq q^\text{Beta}_{\alpha,\beta}(\epsilon)) = \epsilon$
for a random variable $Z$ with probability density function
$\frac{\Gamma(\alpha+\beta)}{\Gamma(\alpha) \Gamma(\beta)} z^{\alpha-1} (1-z)^{\beta-1}$.
The \cite{ClopperPearson1934} confidence intervals we compute are
slightly conservative in practice \citep{Li2009}.

We show in the following lemma
that our particular choice of $f$ defined beforehand satisfies
Condition \ref{condition_intervals} and Condition \ref{condition_probability}
stated in Section \ref{subsection_conditions}.
Therefore, $\epsilon$ is a bound on the probability of having
any false classification.
Other functions $f$, for example based on other spending sequences, can
obviously be used as long as they satisfy
Conditions \ref{condition_intervals} and \ref{condition_probability}.

\begin{lemma}
\label{lemma_clopper_pearson_satisfy_condition}
The confidence intervals computed by the function $f$ satisfy
Conditions \ref{condition_intervals} and \ref{condition_probability}.
\end{lemma}

\begin{proof}
First, we consider an individual Clopper-Pearson confidence
interval $I_i^n$ computed in Algorithm \ref{algorithm_basic}
using $f$ as defined in Appendix \ref{appendix_clopper_pearson}.
To ease notation, we drop the indices $i$ and $n$.

We show that $|I|\leq 2 \xi$,
where  $\xi=\sqrt{ \frac{-1}{2 k} \log \rho }$ and
$\rho=(\eta_k-\eta_{k-\Delta})/(2m)$.
The following probabilities are conditional on $S$ and $k$. 

Suppose $S < k$. Then
the upper limit $p_u$ of the interval $I$ is the solution to
$\PP(N \leq S|p=p_u) = \rho$,
where  $N \sim \text{Binomial}(k,p)$.
If $p > S/k +\xi$  then, by Hoeffding's inequality \citep{Hoeffding1963},
\begin{align*}
\PP(N \leq S)
= \PP \left( \frac{N}{k} - \E \left( \frac{N}{k} \right)
\leq \frac{S}{k} - \E \left( \frac{N}{k} \right) \right)
\leq \exp \left( - \frac{2 (S/k-p)^2 k^2}{k} \right)
< \rho.
\end{align*}
Thus $p_u \leq S/k + \xi$.
If $S=k$ then $p_u=1$, implying $p_u \leq S/k + \xi$.

Similarly, it can be shown that the lower limit $p_l$ of $I$
satisfies $p_l \geq S/k - \xi$.
Hence, $|I| = p_u-p_l \leq 2 \xi$.

Now consider $|f(\cdot,k,\cdot)|$ for
$k \rightarrow \infty$, see Condition \ref{condition_intervals}.
The function $f(S,k,\Delta)$ given in Appendix \ref{appendix_clopper_pearson}
computes Clopper-Pearson confidence intervals
with coverage probability $1-(\eta_k-\eta_{k-\Delta})/m$,
where $\eta_k := \frac{k}{k+r} \epsilon$ for a constant $r>0$.
The sequence $\eta_k$ satisfies
$\eta_{k}-\eta_{k-\Delta} \sim k^{-2}$, implying
$\log( \eta_{k}-\eta_{k-\Delta} ) = o(k)$.
As $|I^n_i|\leq 2\sqrt{ \log ((\eta_k-\eta_{k-\Delta})/(2m)) /(-2 k)}$,
$|I^n_i|\to 0$ as $k\to \infty$.
This proves Condition \ref{condition_intervals}.

Second, we show that the function $f$ given in
Appendix \ref{appendix_clopper_pearson}
computes confidence intervals in such a way that
$\PP(p^\ast_i \in I_i^n ~ \forall i,n) \geq 1-\epsilon$,
thus satisfying Condition \ref{condition_probability}.

Let $k_i^n$ denote the value of $k_i$ in iteration $n$,
and let $\Delta^n$ denote the value of $\Delta$ in iteration $n$,
where $k_i^0 = k_i^1 - \Delta^1 = 0$, $i \in \{ 1,\ldots,m \}$.
The function $f$ defined in Appendix \ref{appendix_clopper_pearson} computes
Clopper-Pearson confidence intervals $I_i^n$ such that
$\PP(p_i^\ast \notin I_i^n) \leq (\eta_{k_i^n}-\eta_{k_i^n-\Delta^n})/m$.

This then yields
\begin{align*}
  \PP(\exists i,n: p^\ast_i \notin I_i^n)
  \leq \sum_{i=1}^m \sum_{n=1}^\infty \PP(p^\ast_i \notin I_i^n)
  \leq \sum_{i=1}^m \sum_{n=1}^\infty (\eta_{k_i^n}-\eta_{k_i^n-\Delta^n})/m
  = \frac{1}{m} \sum_{i=1}^m \epsilon
  = \epsilon,
\end{align*}
using properties of $\eta_k = \frac{k}{k+r} \epsilon$, where $r>0$ is constant.
Condition \ref{condition_probability} is thus satisfied.
\end{proof}

\section{Some properties of multiple testing procedures}
\label{appendix_properties}
In this section we discuss how and under which circumstances
two multiple testing procedures, namely the
\cite{Benjamini1995CFD} procedure
and the \cite{Bonferroni1936} correction,
satisfy the conditions of Theorem \ref{theorem_convergence}.

In this section and the following sections,
$A^c$ denotes the complement of $A \subseteq \{ 1,\ldots,m \}$
with respect to $\{ 1,\ldots,m \}$, where
$m$ is the number of hypotheses under consideration.

\subsection{Properties of the Benjamini-Hochberg procedure}
\label{appendix_benjamini_hochberg}
The False Discovery Rate controlling procedure of
\cite{Benjamini1995CFD}
with threshold $\alpha>0$ is defined as follows.
Given $m$ p-values $p_1,\ldots,p_m$,
their order statistic is denoted by
$p_{(1)} \leq p_{(2)} \leq \ldots \leq p_{(m)}$.
In case of a tie, equal values are assigned a rank in arbitrary
order.
Let $k$ be the largest index $i$ for which
$p_{(i)} \leq \frac{i}{m} \alpha$.
Then, rejecting all the hypotheses corresponding to $p_{(1)},\ldots,p_{(k)}$
ensures that the FDR is at most $\alpha$.
The procedure can  be expressed as
$$h(p) = \left\{ i \in \{ 1,\ldots,m \}: ~\exists j: r_p(j) \geq r_p(i)
\text{ and } m \frac{p_j}{r_p(j)} \leq \alpha \right\},$$
where $r_p(i)$ denotes the rank
of $p_i$ in $p_{(1)} \leq p_{(2)} \leq \ldots \leq p_{(m)}$.

The following theorem states three properties of
the Benjamini-Hochberg procedure $h$
which are slightly stronger than Condition \ref{condition_invariance}.

\begin{theorem}
\label{theorem_benjamini_hochberg}
\begin{enumerate}
  \item $h$ is monotonic.
  \label{theorem_item_monotonicity}
  \item Let $p,q \in [0,1]^m$.
If $q_i \leq \frac{|h(p)| \alpha}{m}$ $\forall i \in h(p)$
and $q_i = p_i$ $\forall i \notin h(p)$, then $h(p)=h(q)$.
  \label{theorem_item_invariance_rejection_area}
  \item Let $p,q \in [0,1]^m$.
If $q_i=p_i$ $\forall i \in h(p)$ and
$q_i > \frac{\alpha}{m}r_p(i)$ $\forall i \notin h(p)$,
then $h(p) = h(q)$.
  \label{theorem_item_invariance_nonrejection_area}
\end{enumerate}
\end{theorem}

\begin{proof}
As $h$ is invariant to permutations, we may assume $p_1 \leq \cdots \leq p_m$.

\ref{theorem_item_monotonicity}.
Let $p \in [0,1]^m$ and $i \in \{ 1,\ldots,m \}$.
It suffices to show that $h(p) \supseteq h(q)$ for any
$q \in [0,1]^m$ given by
$q_j=p_j$ $\forall j \neq i$ and $q_i > p_i$.

Let $k := |h(p)|$ be the largest index which is
rejected when the Benjamini-Hochberg procedure is applied to $p$.
We need to show that $j \notin h(q)$ $\forall j \geq k+1$.

Case 1: $r_{q}(i) \leq k$. This implies $r_{q}(j)=j$ $\forall j \geq k+1$ and hence
$q_j = p_j > \frac{\alpha j}{m} = \frac{\alpha r_{q}(j)}{m}.$
Therefore, $j \notin h(q)$ $\forall j \geq k+1$.

Case 2: $r_{q}(i) \geq k+1$. Let $j \geq k+1$, $j \neq i$.
Then the rank of the $j$th p-value can only drop by one when $p_i$
is replaced by $q_i$, i.e.\ $r_q(j) \in \{ j-1,j \}$.
Thus $q_j = p_j > \frac{\alpha j}{m} \geq \frac{\alpha r_q(j)}{m}$.
Furthermore, as $r_{q}(i) \geq k+1$, $q_i$
takes the position of the former $p_{r_{q}(i)}$ in the ordered sequence of
values from $q$, i.e.\ $q_i \geq p_{r_{q}(i)}$.
Hence, $r_{q}(i) \notin h(p)$
because of $r_{q}(i) \geq k+1$ and thus
$q_i \geq p_{r_{q}(i)} > \frac{\alpha r_{q}(i)}{m}$.
Therefore, $\{ k+1, \ldots, m \} \cup \{ i \} \notin h(q)$.
This proves statement \ref{theorem_item_monotonicity}.

\ref{theorem_item_invariance_rejection_area}.
All $i \notin h(p)$ satisfy $p_i > \frac{|h(p)| \alpha}{m}$, whereas by assumption,
$q_i \leq \frac{|h(p)| \alpha}{m}$ $\forall i \in h(p)$.
Hence, using $q_i = p_i$ $\forall i \notin h(p)$, it follows that
$r_q(i)=r_p(i)$ $\forall i \notin h(p)$. Thus,
$q_i = p_i > \frac{r_p(i) \alpha}{m} = \frac{r_q(i) \alpha}{m}$
for all $i \notin h(p)$. Hence $h(p)^c \subseteq h(q)^c$.

Conversely, define $\tilde{q} := \max \{ q_i: i \in h(p) \}$. As
$\tilde{q} \leq \frac{|h(p)| \alpha}{m} < q_i$
for all $i \notin h(p)$ and as there are exactly $|h(p)|$ values
$q_i \leq \tilde{q}$, the rank of $\tilde{q}$ in $q$ is precisely $|h(p)|$.
As $q_i \leq \tilde{q} \leq \frac{|h(p)| \alpha}{m}$
$\forall i \in h(p)$,
all $\{ q_i \}_{i \in h(p)}$ are rejected and $h(p) \subseteq h(q)$.
This proves statement \ref{theorem_item_invariance_rejection_area}.

\ref{theorem_item_invariance_nonrejection_area}.
As $q_i = p_i$ for all $i \in h(p)$, have $h(p) \subseteq h(q)$.

Let $i \notin h(p)$. If $r_q(i) \leq r_p(i)$, then
$q_i > \frac{\alpha}{m}r_p(i) \geq \frac{\alpha}{m}r_q(i)$.
If $r_q(i) > r_p(i)$, $q_i$ replaces a
$q_j > \frac{\alpha}{m}r_p(j)$ at rank $r_p(j)$
in the sorted sequence of $q$, hence $r_q(i)=r_p(j)$ and
$q_i \geq q_j > \frac{\alpha}{m}r_p(j) = \frac{\alpha}{m}r_q(i)$.
Thus $q_i > \frac{\alpha}{m}r_q(i)$ $\forall i \notin h(p)$,
which implies $h(p)^c \subseteq h(q)^c$.
This proves statement \ref{theorem_item_invariance_nonrejection_area}.
\end{proof}

The second statement
shows that all the p-values in the set of rejections can be
increased up to a certain bound without affecting the result of $h$.
The third statement states
that $h$ stays invariant if
the p-values in the non-rejection area
are replaced by arbitrary values above the Benjamini-Hochberg line
(see Example \ref{example_visualization}).

\begin{corollary}
\label{corollary_benjamini_hochberg}
$h$ satisfies Condition \ref{condition_invariance}.
\end{corollary}

\begin{proof}
Statement $1$ of Condition \ref{condition_invariance}
is satisfied as $h$ is monotonic by Theorem
\ref{theorem_benjamini_hochberg}.

To prove that the Benjamini-Hochberg procedure $h$ also satisfies
the second statement of Condition \ref{condition_invariance},
it suffices to show that for $p,q \in [0,1]^m$, both
$q_i \leq p_i$ $\forall i \in h(p)$ and $q_i = p_i$ $\forall i \notin h(p)$
as well as
$q_i = p_i$ $\forall i \in h(p)$ and $q_i \geq p_i$ $\forall i \notin h(p)$
imply $h(p) = h(q)$.

Indeed, let $p,q \in [0,1]^m$ be such that
$q_i \leq p_i$ $\forall i \in h(p)$ and $q_i = p_i$ $\forall i \notin h(p)$.
We have $p_i \leq \frac{|h(p)|\alpha}{m}$ $\forall i \in h(p)$
by definition of $h$, thus $q_i \leq p_i \leq \frac{|h(p)|\alpha}{m}$
$\forall i \in h(p)$ and
$h(p) = h(q)$ by statement \ref{theorem_item_invariance_rejection_area}
of Theorem \ref{theorem_benjamini_hochberg}.

Similarly, let $p,q \in [0,1]^m$ be such that
$q_i = p_i$ $\forall i \in h(p)$ and $q_i \geq p_i$ $\forall i \notin h(p)$.
Using that $p_i > \frac{\alpha}{m}r_p(i)$ $\forall i \notin h(p)$ it
immediately follows that $q_i \geq p_i>\frac{\alpha}{m}r_p(i)$ $\forall i \notin h(p)$
and thus $h(p) = h(q)$
by statement \ref{theorem_item_invariance_nonrejection_area}
of Theorem \ref{theorem_benjamini_hochberg}.
\end{proof}

The next lemma states that $h$ is locally constant
for almost all arguments:

\begin{lemma}
\label{lemma_invariance_delta}
If $p^\ast \in [0,1]^m$ with $p^\ast_{(i)} \neq i \alpha/m$,
$i \in \{ 1,\ldots,m \}$, then there exists $\delta>0$ such that
$p \in [0,1]^m$ and
$\left\Vert p-p^\ast \right\Vert < \delta$
imply $h(p^\ast)=h(p)$.
\end{lemma}

\begin{proof}
The function $h$ stays invariant if all p-values do not
change their rank outside of a tie and if
no p-value crosses the Benjamini-Hochberg threshold line.

As $h$ is invariant to permutations,
we may assume $p^\ast_1 \leq \cdots \leq p^\ast_m$.

Let
$\delta := \min \left(
\left\{ \frac{p^\ast_i-p^\ast_{i-1}}{2}: i=2,\ldots,m
\text{ with } p^\ast_{i-1}<p^\ast_i \right\} \cup
\left\{ |p^\ast_i - \frac{i\alpha}{m}|: i=1,\ldots,m \right\} \right)$.

Let $p \in [0,1]^m$ with $\left\Vert p-p^\ast \right\Vert < \delta$.
Then $p^\ast_{i-1}<p^\ast_i$ implies
$p_{i-1} < p^\ast_{i-1}+\delta \leq p^\ast_i-\delta < p_i$.
Thus, by possibly permuting indices corresponding to tied
values in $p$, we may assume
$p^\ast_1 \leq \cdots \leq p^\ast_m$
and $p_1 \leq \cdots \leq p_m$. The ranks of the p-values
in $p^\ast$ and $p$ are therefore the same.

Furthermore, $|p_i-p^\ast_i| < \delta \leq |p^\ast_i-i\alpha/m|$
for all $i \in \{ 1,\ldots,m \}$, implying that
$p^\ast_i$ and $p_i$ lie on the same side of the
Benjamini-Hochberg line. Hence, $h(p^\ast)=h(p)$.
\end{proof}

Lemma \ref{lemma_invariance_delta} thus shows that
the condition on $p^\ast$ in Theorem \ref{theorem_convergence}
is satisfied for all the p-values except for those
lying exactly on the Benjamini-Hochberg line.

\subsection{Properties of the Bonferroni correction}
\label{appendix_bonferroni_correction}
The \cite{Bonferroni1936} correction
controls the Familywise Error Rate, defined
by $\text{FWER}:=\PP(V \geq 1)$, where $V$ is the number of hypotheses from
the null which have been rejected (false positives).
The method tests all $m$ hypotheses $H_{01},\ldots,H_{0m}$
at threshold $\alpha/m$ to guarantee $\text{FWER} \leq \alpha$.
The Bonferroni correction $h_B$ returning the set of rejected indices
can be stated as
$$h_B(p) = \left\{ i \in \{ 1,\ldots,m \}: p_i \leq \alpha/m \right\}.$$
Similarly to Theorem \ref{theorem_benjamini_hochberg},
the following theorem
states two key properties of $h_B$ which are slightly
stronger than the corresponding statements of Condition
\ref{condition_invariance}.

\begin{theorem}
\label{theorem_Bonferroni}
\begin{enumerate}
  \item $h_B$ is monotonic.
  \label{theorem_Bonferroni_item_monotonicity}
  \item Let $p,q \in [0,1]^m$. If $q_i \leq \frac{\alpha}{m}$ $\forall i \in h_B(p)$
and $q_i > \frac{\alpha}{m}$ $\forall i \notin h_B(p)$, then $h_B(p)=h_B(q)$.
  \label{theorem_Bonferroni_item_invariance_area}
\end{enumerate}
\end{theorem}

\begin{proof}
\ref{theorem_Bonferroni_item_monotonicity}.
Let $p \in [0,1]^m$ and $i \in \{ 1,\ldots,m \}$.
It suffices to show that $h_B(p) \supseteq h_B(q)$ for any
$q \in [0,1]^m$ given by
$q_j=p_j$ $\forall j \neq i$ and $q_i > p_i$.

If $p_i > \alpha/m$, then $i$ is non-rejected. Increasing $p_i$ even
further will thus not change the result of $h_B$ as $q_i > p_i > \alpha/m$.
Therefore, $h_B(q)=h_B(p)$.

If $p_i \leq \alpha/m$ and $q_i>p_i$, the result of $h_B(q)$ depends on
whether $q_i$ is greater than $\alpha/m$ or not. In the first case, $i$
is non-rejected in $h_B(q)$, thus $h_B(q) \subset h_B(p)$. In the second
case, $q_i \leq \alpha/m$ is still rejected and thus $h_B(q)=h_B(p)$.
This proves statement \ref{theorem_Bonferroni_item_monotonicity}.

\ref{theorem_Bonferroni_item_invariance_area}.
All $q_i \leq \alpha/m$, $i \in h_B(p)$,
are rejected, thus $h_B(p) \subseteq h_B(q)$.
Similarly, all $q_i > \alpha/m$, $i \notin h_B(p)$,
are non-rejected, thus $h_B(p)^c \subseteq h_B(q)^c$.
This proves statement \ref{theorem_Bonferroni_item_invariance_area}.
\end{proof}

The second statement of Theorem \ref{theorem_Bonferroni}
shows that the result of $h_B$ is not affected
if p-values in the rejection (non-rejection) area are
replaced by arbitrary values below (above)
the constant testing threshold $\alpha/m$.

\begin{corollary}
\label{corollary_Bonferroni}
$h_B$ satisfies Condition \ref{condition_invariance}.
\end{corollary}

\begin{proof}
Statement $1$ of Condition \ref{condition_invariance}
is satisfied as $h$ is monotonic by Theorem \ref{theorem_Bonferroni}.

Statement \ref{theorem_Bonferroni_item_invariance_area}
of Theorem \ref{theorem_Bonferroni}
shows that the Bonferroni
correction also satisfies the second statement of Condition
\ref{condition_invariance}.
Indeed, let $p,q \in [0,1]^m$
be given such that $q_i \leq p_i$ $\forall i \in h(p)$
and $q_i \geq p_i$ $\forall i \notin h(p)$.
By definition of $h_B$, $p_i \leq \frac{\alpha}{m}$
$\forall i \in h_B(p)$ and $p_i > \frac{\alpha}{m}$
$\forall i \notin h_B(p)$.
In particular, $q_i \leq p_i \leq \frac{\alpha}{m}$
$\forall i \in h_B(p)$ and $q_i \geq p_i > \frac{\alpha}{m}$
$\forall i \notin h_B(p)$, thus $h_B(p)=h_B(q)$
by statement \ref{theorem_Bonferroni_item_invariance_area}
of Theorem \ref{theorem_Bonferroni}.
\end{proof}

Similarly to Lemma \ref{lemma_invariance_delta}, the Bonferroni correction is
locally constant for almost all values:

\begin{lemma}
\label{lemma_Bonferroni_invariance_delta}
For all $p^\ast \in \left( [0,\alpha/m) \cup (\alpha/m,1] \right)^m$
there exists $\delta>0$ such that
$p \in [0,1]^m$ and
$\left\Vert p-p^\ast \right\Vert < \delta$
imply $h_B(p^\ast)=h_B(p)$.
\end{lemma}

\begin{proof}
Let $\delta:=\min_{i \in \{ 1,\ldots,m \}}|p^\ast_i - \alpha/m|$.
Let $p \in [0,1]^m$ with $\left\Vert p-p^\ast \right\Vert < \delta$.
For all $i \in \{ 1,\ldots,m \}$,
this implies that $p_i$ and $p^\ast_i$ lie on the same side
of the threshold $\alpha/m$.
Therefore, $h_B(p^\ast)=h_B(p)$.
\end{proof}

Lemma \ref{lemma_Bonferroni_invariance_delta} thus shows that
the condition on $p^\ast$ in Theorem \ref{theorem_convergence}
is satisfied for all the p-values except for those
lying exactly on the threshold $\alpha/m$.

\section{Proofs of Section 3}
\label{appendix_proofs}

\begin{proof}[Proof of Lemma \ref{lemma_monotonicity_sets}]
1. By construction of Algorithm \ref{algorithm_basic},
all the confidence intervals $I_i^n$ are nested.
Therefore, the sequence $(\max I_i^n)_{i=1,\ldots,m}$ is decreasing in $n$.
Thus,
$$\underline{A}_n = h((\max I_i^n)_{i=1,\ldots,m})
\subseteq h((\max I_i^{n+1})_{i=1,\ldots,m}) = \underline{A}_{n+1}$$
by monotonicity of $h$.
Similarly, $\overline{A}_n \supseteq \overline{A}_{n+1}$
as $(\min I_i^n)_{i=1,\ldots,m}$ is increasing in $n$.

2. Given that $p_i^\ast \in I_i^n$ $\forall i,n$,
the ideal p-values satisfy $p_i^\ast \leq \max I_i^n$ $\forall i,n$.
When applied to the vectors
$(p_i^\ast)_{i=1,\ldots,m} \leq (\max I_i^n)_{i=1,\ldots,m}$,
the monotonicity of $h$ yields
$$\A = h((p_i^\ast)_{i=1,\ldots,m}) \supseteq h((\max I_i^n)_{i=1,\ldots,m}) = \underline{A}_n.$$
Similarly, $\A \subseteq \overline{A}_n$
as $(p_i^\ast)_{i=1,\ldots,m} \geq (\min I_i^n)_{i=1,\ldots,m}$.
\end{proof}

\begin{proof}[Proof of Theorem \ref{theorem_convergence}]
Let $B_n = \overline{A}_n \setminus \underline{A}_n$.
Suppose $\exists i \in \lim \sup_{n \rightarrow \infty} B_n$.
By construction of Algorithm \ref{algorithm_basic},
$n \rightarrow \infty$ implies $k_i \rightarrow \infty$.
Condition \ref{condition_intervals} thus yields
$|I_i^n| = |f(\cdot,k_i,\cdot)|\rightarrow 0$ as $n \rightarrow \infty$.
Let $\delta$ be as given in the theorem.
As $B_n \subseteq \{1,\ldots,m \}$ is finite $\forall n \in \N$, there exists
$n_0 \in \N$ such that $|I_i^n|^2 < \delta^2/m$ for
$n \geq n_0$ and all $i \in \lim \sup_{n \rightarrow \infty} B_n$.

We show that for all $n \geq n_0$,
$$h( (\min I_i^n)_{i \in \{ 1,\ldots,m \}} ) = h(p^\ast) =
h( (\max I_i^n)_{i \in \{ 1,\ldots,m \}} ).$$
To do this, we show $h(p^{(j)})=h(p^{(j+1)})$, $j \in \{ 1,\ldots,6 \}$, where
\begin{center}
\begin{tabular}{lll}
$p^{(1)} := (\min I_i^n)_{i \in \{ 1,\ldots,m \}},$&
$p^{(4)} := p^\ast,$\\
$p^{(2)} := \begin{cases} \begin{matrix} \min I_i^n & i \in \overline{A}_n,\\ p_i^\ast & i \notin \overline{A}_n, \end{matrix} \end{cases}$&
$p^{(5)} := \begin{cases} \begin{matrix} \max I_i^n & i \in B_n,\\ p_i^\ast & i \notin B_n, \end{matrix} \end{cases}$\\
$p^{(3)} := \begin{cases} \begin{matrix} \min I_i^n & i \in B_n,\\ p_i^\ast & i \notin B_n, \end{matrix} \end{cases}$&
$p^{(6)} := \begin{cases} \begin{matrix} \max I_i^n & i \in \overline{A}_n,\\ p_i^\ast & i \notin \overline{A}_n, \end{matrix} \end{cases}$
\end{tabular}
\end{center}
and $p^{(7)} := (\max I_i^n)_{i \in \{ 1,\ldots,m \}}$.
The following holds true on the event $\{ p_i^\ast \in I^n_i ~\forall i,n \}$.

(1) By definition, $\overline{A}_n = h(p^{(1)})$.
As $p^{(2)}_j = p_j^\ast \geq \min I_j^n = p^{(1)}_j$ $\forall j \notin \overline{A}_n$
and $p^{(2)}_j = p^{(1)}_j$ $\forall j \in \overline{A}_n$,
the second statement of Condition \ref{condition_invariance}
yields $\overline{A}_n = h(p^{(1)}) = h(p^{(2)})$.

(2) As $(\max I_i^n)_{i \in \{ 1,\ldots,m \}} \geq p^{(3)}$ and as
$h$ is monotonic, $\underline{A}_n \subseteq h(p^{(3)})$.
As $p^{(2)}_j = \min I_j^n \leq p_j^\ast = p^{(3)}_j$ $\forall j \in \underline{A}_n$
and $p^{(2)}_j = p^{(3)}_j$ $\forall j \notin \underline{A}_n$,
the second statement of Condition \ref{condition_invariance}
yields $h(p^{(2)}) = h(p^{(3)})$.

(3) For all $n \geq n_0$ and all $i \in \lim \sup_{n \rightarrow \infty} B_n$,
$|I_i^n|^2 < \delta^2/m$ implies
$\| p^{(3)} - p^\ast \| < \delta$ and hence
$h(p^{(3)}) = h(p^{(4)}) = h(p^\ast)$ $\forall n \geq n_0$
by definition of $\delta$ in the theorem.

Arguing similarly to (1), (2), (3) we can show
$h(p^{(4)}) = h(p^{(5)})$, $h(p^{(5)}) = h(p^{(6)})$
and $h(p^{(6)}) = h(p^{(7)}) = \underline{A}_n$.
\end{proof}

\begin{proof}[Proof of Corollary \ref{corollary_convergence}]
By Theorem \ref{theorem_convergence}
we have $\underline{A}_n \rightarrow \A$, $\overline{A}_n \rightarrow \A$
as $n \rightarrow \infty$
conditional on $\{ p_i^\ast \in I^n_i ~\forall i,n \}$.
Under Condition \ref{condition_probability},
this event occurs with probability
$\PP(p_i^\ast \in I_i^n ~ \forall i,n) \geq 1-\epsilon$, hence
$\PP(\underline{A}_n \rightarrow \A, \overline{A}_n \rightarrow \A) \geq 1-\epsilon$.
\end{proof}

% bibliography


\begin{thebibliography}{}
\bibitem[{Benjamini \& Hochberg(1995)}]{Benjamini1995CFD}
Benjamini, Y. \& Hochberg, Y. (1995). {Controlling the false discovery rate: A
  practical and powerful approach to multiple testing}. {\em J. R. Stat. Soc.
  Ser. B Stat. Methodol.\/} {\bf 57}, 289--300.

\bibitem[{Besag \& Clifford(1991)}]{BesagClifford1991}
Besag, J. \& Clifford, P. (1991). {Sequential Monte Carlo p-values}. {\em
  Biometrika\/} {\bf 78}, 301--304.

\bibitem[{Bonferroni(1936)}]{Bonferroni1936}
Bonferroni, C. (1936). {Teoria statistica delle classi e calcolo delle
  probabilit\`a}. {\em Pubblicazioni del R Istituto Superiore di Scienze
  Economiche e Commerciali di Firenze\/} {\bf 8}, 3--62.

\bibitem[{Clopper \& Pearson(1934)}]{ClopperPearson1934}
Clopper, C. \& Pearson, E. (1934). {The use of confidence or fiducial limits
  illustrated in the case of the binomial}. {\em Biometrika\/} {\bf 26},
  404--413.

\bibitem[{Cohen et~al.(2012)Cohen, Ashkenazy, Burstein \& Pupko}]{Cohen2012}
Cohen, O., Ashkenazy, H., Burstein, D. \& Pupko, T. (2012). {Uncovering the
  co-evolutionary network among prokaryotic genes}. {\em Bioinformatics\/} {\bf
  28}, i389--i394.

\bibitem[{Farcomeni(2007)}]{Farcomeni2007}
Farcomeni, A. (2007). {Some Results on the Control of the False Discovery Rate
  under Dependence}. {\em Scand. J. Stat.\/} {\bf 34}, 275--297.

\bibitem[{Farcomeni(2009)}]{Farcomeni2009}
Farcomeni, A. (2009). {Generalized Augmentation to Control the False Discovery
  Exceedance in Multiple Testing}. {\em Scand. J. Stat.\/} {\bf 36}, 501--517.

\bibitem[{Finner et~al.(2012)Finner, Gontscharuk \& Dickhaus}]{Finner2012}
Finner, H., Gontscharuk, V. \& Dickhaus, T. (2012). {False Discovery Rate
  Control of Step-Up-Down Tests with Special Emphasis on the Asymptotically
  Optimal Rejection Curve}. {\em Scand. J. Stat.\/} {\bf 39}, 382--397.

\bibitem[{Gandy(2009)}]{gandy09:SeqImplMC}
Gandy, A. (2009). {Sequential implementation of Monte Carlo tests with
  uniformly bounded resampling risk}. {\em J. Amer. Statist. Assoc.\/} {\bf
  104}, 1504--1511.

\bibitem[{Gandy \& Rubin-Delanchy(2013)}]{GandyRubinDelanchy2013}
Gandy, A. \& Rubin-Delanchy, P. (2013). {An algorithm to compute the power of
  Monte Carlo tests with guaranteed precision}. {\em Ann. Statist.\/} Accepted
  for publication.

\bibitem[{Gleser(1996)}]{Gleser1996}
Gleser, L. (1996). {Comment on 'Bootstrap Confidence Intervals' by T. J.
  DiCiccio and B. Efron}. {\em Statist. Sci.\/} {\bf 11}, 219--221.

\bibitem[{Guo \& Peddada(2008)}]{GuoPedadda2008}
Guo, W. \& Peddada, S. (2008). {Adaptive Choice of the Number of Bootstrap
  Samples in Large Scale Multiple Testing}. {\em Stat. Appl. Genet. Mol.
  Biol.\/} {\bf 7}, 1--16.

\bibitem[{Gusenleitner et~al.(2012)Gusenleitner, Howe, Bentink, Quackenbush \&
  Culhane}]{Gusenleitner2012}
Gusenleitner, D., Howe, E., Bentink, S., Quackenbush, J. \& Culhane, A. (2012).
  {iBBiG: iterative binary bi-clustering of gene sets}. {\em Bioinformatics\/}
  {\bf 28}, 2484--2492.

\bibitem[{Hoeffding(1963)}]{Hoeffding1963}
Hoeffding, W. (1963). Probability inequalities for sums of bounded random
  variables. {\em J. Amer. Statist. Assoc.\/} {\bf 58}, 13--30.

\bibitem[{Jiang \& Salzman(2012)}]{JiangSalzman2012}
Jiang, H. \& Salzman, J. (2012). {Statistical properties of an early stopping
  rule for resampling-based multiple testing}. {\em Biometrika\/} {\bf 99},
  973--980.

\bibitem[{Jiao \& Zhang(2010)}]{JiaoZhang2010}
Jiao, S. \& Zhang, S. (2010). {A mixture model based approach for estimating
  the FDR in replicated microarray data}. {\em Journal of Biomedical Science
  and Engineering\/} {\bf 20}, 317--321.

\bibitem[{Knijnenburg et~al.(2009)Knijnenburg, Daran, van~den Broek,
  Daran-Lapujade, de~Winde, Pronk, Reinders \& Wessels}]{Knijnenburg2009}
Knijnenburg, T., Daran, J.-M., van~den Broek, M., Daran-Lapujade, P., de~Winde,
  J., Pronk, J., Reinders, M. \& Wessels, L. (2009). {Combinatorial effects of
  environmental parameters on transcriptional regulation in Saccharomyces
  cerevisiae: A quantitative analysis of a compendium of chemostat-based
  transcriptome data}. {\em BMC Genomics\/} {\bf 10}.

\bibitem[{Lage-Castellanos et~al.(2010)Lage-Castellanos, Mart\'inez-Montes,
  Hern\'andez-Cabrera \& Gal\'an}]{LageCastellanos2010}
Lage-Castellanos, A., Mart\'inez-Montes, E., Hern\'andez-Cabrera, J. \&
  Gal\'an, L. (2010). {False discovery rate and permutation test: An evaluation
  in ERP data analysis}. {\em Stat. Med.\/} {\bf 29}, 63--74.

\bibitem[{Li et~al.(2009)Li, Tai \& Nott}]{Li2009}
Li, J., Tai, B. \& Nott, D. (2009). {Confidence interval for the bootstrap
  P-value and sample size calculation of the bootstrap test}. {\em J.
  Nonparametr. Stat.\/} {\bf 21}, 649--661.

\bibitem[{Meinshausen(2006)}]{Meinshausen2006}
Meinshausen, N. (2006). {False Discovery Control for Multiple Tests of
  Association Under General Dependence}. {\em Scand. J. Stat.\/} {\bf 33},
  227--237.

\bibitem[{Nusinow et~al.(2012)Nusinow, Kiezun, O'Connell, Chick, Yue, Maas,
  Gygi \& Sunyaev}]{Nusinow2012}
Nusinow, D., Kiezun, A., O'Connell, D., Chick, J., Yue, Y., Maas, R., Gygi, S.
  \& Sunyaev, S. (2012). {Network-based inference from complex proteomic
  mixtures using SNIPE}. {\em Bioinformatics\/} {\bf 28}, 3115--3122.

\bibitem[{Pekowska et~al.(2010)Pekowska, Benoukraf, Ferrier \&
  Spicuglia}]{Pekowska2010}
Pekowska, A., Benoukraf, T., Ferrier, P. \& Spicuglia, S. (2010). {A unique
  H3K4me2 profile marks tissue-specific gene regulation}. {\em Genome
  Research\/} {\bf 20}, 1493--1502.

\bibitem[{Pounds \& Cheng(2006)}]{PoundsCheng2006}
Pounds, S. \& Cheng, C. (2006). {Robust estimation of the false discovery
  rate}. {\em Bioinformatics\/} {\bf 22}, 1979--1987.

\bibitem[{Rahmatallah et~al.(2012)Rahmatallah, Emmert-Streib \&
  Glazko}]{Rahmatallah2012}
Rahmatallah, Y., Emmert-Streib, F. \& Glazko, G. (2012). {Gene set analysis for
  self-contained tests: complex null and specific alternative hypotheses}. {\em
  Bioinformatics\/} {\bf 28}, 3073--3080.

\bibitem[{Sandve et~al.(2011)Sandve, Ferkingstad \& Nygard}]{Sandve2011}
Sandve, G., Ferkingstad, E. \& Nygard, S. (2011). {Sequential Monte Carlo
  multiple testing}. {\em Bioinformatics\/} {\bf 27}, 3235--3241.

\bibitem[{Tamhane \& Liu(2008)}]{TamhaneLiu2008}
Tamhane, A. \& Liu, L. (2008). {On weighted Hochberg procedures}. {\em
  Biometrika\/} {\bf 95}, 279--294.

\bibitem[{van Wieringen et~al.(2008)van Wieringen, van~de Wiel \& van~der
  Vaart}]{Wieringen2008}
van Wieringen, W., van~de Wiel, M. \& van~der Vaart, A. (2008). {A Test for
  Partial Differential Expression}. {\em J. Amer. Statist. Assoc.\/} {\bf 103},
  1039--1049.

\bibitem[{Westfall \& Troendle(2008)}]{WestfallTroendle2008}
Westfall, P. \& Troendle, J. (2008). {Multiple Testing with Minimal
  Assumptions}. {\em Biom. J.\/} {\bf 50}, 745--755.

\bibitem[{Westfall \& Young(1993)}]{WestfallYoung1993}
Westfall, P. \& Young, S. (1993). {\em {Resampling-based multiple testing:
  Examples and methods for p-value adjustment}\/}. Wiley, New York.
\end{thebibliography}
\end{document}